\documentclass[10pt, conference, letterpaper]{IEEEtran}
\makeatletter
\def\ps@headings{%
\def\@oddhead{\mbox{}\scriptsize\rightmark \hfil \thepage}%
\def\@evenhead{\scriptsize\thepage \hfil \leftmark\mbox{}}%
\def\@oddfoot{}%
\def\@evenfoot{}}
\makeatother
\ifCLASSINFOpdf
 \usepackage[pdftex]{graphicx}
\else

\fi

\usepackage{float}
\usepackage{subfig}
\usepackage{amsmath}
\usepackage{algorithm}
\usepackage[noend]{algpseudocode}
\usepackage{array}
\usepackage{mdwmath}
\usepackage{eqparbox}
\usepackage[mathscr]{euscript}
\usepackage{fixltx2e}
\usepackage{url}
\usepackage{amsfonts}
\usepackage[svgnames]{xcolor}

\usepackage{tikz}
\usetikzlibrary{arrows,
                calc,
                decorations.pathreplacing,
                calligraphy,
                matrix,
                positioning
                }

\newtheorem{definition}{Definition}
\newtheorem{proposition}{Proposition}
\newtheorem{corollary}{Corollary}

\newtheorem{lemma}{Lemma}

\newcommand{\alphabet}{\mathcal{A}}

\newcommand{\nodes}{\V}
\newcommand{\AD}{\mathcal{F}}

\newcommand{\V}{\mathcal{V}}
\newcommand{\E}{\mathcal{E}}

\newcommand{\Gg}{\mathcal{G}}

\newcommand{\N}{\mathcal{N}}
\renewcommand{\P}{\mathcal{P}}

\newcommand{\ca}{\lambda} 
\renewcommand{\P}{\mathcal{P}}   
\newcommand{\diam}{\text{diam }}
\newcommand{\voisin}[1]{Nei(#1)}  
\newcommand{\send}[3]{Send(#1,#2,#3)}
 
\newcommand{\tr}{\mathcal{T}}   
\newcommand{\pile}{H}   
\newcommand{\topp}[1]{Top(#1)}  
\newcommand{\toppp}[1]{Top^{-1}{(#1)}}  
\newcommand{\haut}[1]{h(#1)}  

\IEEEoverridecommandlockouts

\begin{document}

%

%
%
%

\author{\IEEEauthorblockN{Mohamed~Lamine~Lamali\IEEEauthorrefmark{1}, Simon~Lassourreuille\IEEEauthorrefmark{1}, Stephan~Kunne\IEEEauthorrefmark{2}, Johanne~Cohen\IEEEauthorrefmark{2}
\IEEEauthorblockA{\IEEEauthorrefmark{1}LaBRI-CNRS. Universit\'e de Bordeaux. France.}
\IEEEauthorblockA{\IEEEauthorrefmark{2}LRI-CNRS. Universit\'e Paris-Sud, Universit\'e Paris Saclay. France.}
}
\url{mohamed_lamine.lamali@u-bordeaux.fr }\ \ \
\url{johanne.cohen@lri.fr}
}

\title{A stack-vector routing protocol for automatic tunneling\thanks{This paper will appear in the proceedings of IEEE~INFOCOM~2019.}}

\maketitle

\begin{abstract}
In a network, a tunnel is a part of a path where a protocol is encapsulated in another one. A tunnel starts with an encapsulation and ends with the corresponding decapsulation. Several tunnels can be nested at some stage, forming a protocol stack. Tunneling is very important nowadays and it is involved in several tasks: IPv4/IPv6 transition, VPNs, security (IPsec, onion routing), etc. However, tunnel establishment is mainly performed manually or by script, which present obvious scalability issues. Some works attempt to automate a part of the process (e.g., TSP, ISATAP, etc.). However, the determination of the tunnel(s) endpoints is not fully automated, especially in the case of an arbitrary number of nested tunnels. The lack of routing protocols performing automatic tunneling is due to the unavailability of path computation algorithms taking into account encapsulations and decapsulations. There is a polynomial centralized algorithm to perform the task. However, to the best of our knowledge, no fully distributed path computation algorithm is known. 
Here, we propose the first fully distributed algorithm for path computation with automatic tunneling, i.e., taking into account encapsulation, decapsulation and conversion of protocols. Our algorithm is a generalization of the distributed Bellman-Ford algorithm, where the distance vector is replaced by a protocol stack vector. This allows to know how to route a packet with some protocol stack. We prove that the messages size of our algorithm is polynomial, even if the shortest path can be of exponential length. We also prove that the algorithm converges after a polynomial number of steps in a synchronized setting. 
We adapt our algorithm into a \textit {proto}-protocol for routing with automatic tunneling and we show its efficiency through simulations.
\end{abstract}


\section{Introduction}
\label{sec:inro}

Routing is one of the most important tasks in any network, and particularly in the Internet. Routing protocols are generally (distributed) versions of path computation algorithms in graphs. For instance, RIP relies on the distributed Bellman-Ford algorithm, while OSPF  uses Dijkstra's algorithm. 
These routing protocols were developed in the early years of the Internet. Thus, they work only in networks using the same \textit{protocol}\footnote{In this paper, \textit{communication protocols} and \textit{routing protocols} should not be confused. For simplicity, we refer to the first ones as \textit{protocols}, and to the second ones explicitly as \textit{routing protocols}.} (generally IP). However, Internet encompasses now several protocols. For instance, IPv4 and IPv6 coexist, and a lot of other situations involve interoperability of different protocols within the Internet. Thus, a path between a source and a destination may contain several portions using different protocols (for example, if the path crosses several domains or Autonomous Systems).
The mapping from a protocol to another one along a path is generally done in two ways: i) \textit{Conversion}: a packet of some protocol is \textit{converted} into a packet of another one (e.g., NAT-TP~\cite{RFC2766}); ii) \textit{Encapsulation}: a packet of some protocol is \textit{encapsulated} (or nested) in a packet of another one, thus being transparent to the crossed nodes until its \textit{decapsulation} (the reverse operation). All these operations are called \textit{adaptation functions} further in the paper.

A tunnel is a subpath beginning at an encapsulation and ending at the corresponding decapsulation. Tunneling is a ubiquitous  feature in the Internet nowadays. It is involved in virtual networks, security (an encryption can be modeled as an encapsulation), interoperability, etc. Tunnels can be nested to achieve different tasks. In this context, topology connectivity is not enough to ensure communication. Even if the network is (strongly) connected, there may not exist a suitable adaptation function sequence to reach some destination from some node. If a path has such a suitable sequence, it is said to be \textit{feasible}.
However, current routing protocols are unable to handle protocol heterogeneity. Nowadays, the endpoints of tunnels are manually configured or determined from a precomputed list. Several works try to automate a part of the process (e.g., TSP~\cite{rfc5572}, ISATAP~\cite{rfc5572}, etc.). They are referenced as automatic tunneling. However, even if the negotiation of some parameters  of the tunnel establishment are automated, such as tunnel type and DNS registration, the determination of the tunnel endpoints is not fully automated. Moreover, these approaches cannot handle an arbitrary number of nested tunnels.
The main reason is that the underlying algorithms of routing protocols cannot handle adaptation functions. The authors of~\cite{lamali2018algo} propose a centralized polynomial algorithm to solve path computation problem in a multiple protocol context. However, their approach is based on language theory, and is unlikely to be distributed. To the best of our knowledge, no distributed algorithm for fully automated tunneling, i.e., with automatic determination of the (possibly nested) tunnel endpoints was proposed before this work.

Our goal is to design a distributed algorithm for path computation with automatic tunneling. This corresponds to the classical \textit{All-Pairs Shortest Path} (APSP) problem, but in a network involving adaptation functions. To achieve this goal, we generalize the distributed Bellman-Ford algorithm, in order to take into account possible encapsulations, and to automatically establish (nested) tunnels. The \textit{distance-vector} is replaced by a \textit{stack-vector}. The route to follow for a packet is then determined by its destination and its protocol stack. 
\subsection*{Our contributions:}
\begin{enumerate}
	\item We design the first distributed algorithm for routing with automatic tunneling;
	\item We prove nontrivial bounds on the maximum protocol stack height and on the algorithm convergence: we show that the maximum protocol stack height (number of encapsulated protocols at the same time) of the shortest path (involving tunnels) between two nodes is at most $\ca n^2$, where $n$ is the number of nodes and $\ca$ is the number of protocols in the network. This implies that the maximum message size is also polynomial, despite the fact that the shortest path may be of superpolynomial length in this context;
	\item We design a \textit{proto}-protocol for routing with automatic tunneling and we evaluate its efficiency through simulations.
\end{enumerate}

The paper is organized as follows: Section~\ref{sec:problem} describes the problem and discusses the related work. Section~\ref{sec:model} details the model used through this paper and formalizes the notion of path feasibility. Section~\ref{sec:algorithm} describes the proposed algorithm and study its convergence and its message size. Section~\ref{sec:proto} describes a basic implementation of a routing protocol based on our algorithm, while Section~\ref{sec:simuls} presents the simulation results. Finally, Section~\ref{sec:conclusion} concludes the paper.

\section{The problem}
\label{sec:problem}

\subsection{Problem definition and illustration} 
We aim to illustrate the path computation problem through an example. Figure~\ref{fig:with_without} depicts a network with $8$ nodes\footnote{Note that, in this example, the underlying topology is not a symmetric directed graph, while in our model (see Section~\ref{sec:model}) the directed graph must be symmetric. In this context, a symmetric directed graph is defined as a graph where a link $(U,V)$ exists if and only if the reverse link $(V,U)$ exists. However, the example on Figure~\ref{fig:with_without} is simplified for readability purpose. It can be easily converted into a symmetric directed graph with the same properties.}. We want to compute a path from node $S$ to node $D$. Both of them use protocol $a$. However, the nodes connecting $S$ and $D$ use different protocols. And some nodes are able to map a protocol to another one. A packet of any protocol comprises a \textit{header} field and a \textit{data} field. The structure of the header field is specific to a protocol. Converting a packet of protocol $a$ to protocol $b$ consists in rewriting the header and converting all the information (for example, the node identifiers, the routing options, etc.). Encapsulating a packet of $a$ in $b$ consists in considering the whole packet (header and data) of $a$ as the data field of $b$, and then adding a header of protocol $b$ at the top of the packet.
For example, node $U_1$ is only able to encapsulate (a packet of) protocol $a$ in protocol $b$, while node $U_2$ is only able to receive and send protocol $b$ without any change of the possible encapsulated protocols (which are transparent to it). The capabilities (adaptation functions) of each node are listed above it. Each step of the path is associated with a \textit{protocol stack}, i.e., the current sequence of encapsulated protocols, the current protocol being at the top of the stack. In Figure~\ref{fig:without}, node $S$ emits (a packet of) protocol $a$, so at this stage the protocol stack contains only $a$. When the packet crosses node $U_1$ it is encapsulated in protocol $b$. Thus, the current protocol stack is $ab$ (from bottom to top). At $U_2$ the protocol stack remains the same, since node $U_2$ passively retransmits packets with current protocol $b$. Node $U_4$ then decapsulates protocol $a$ from $b$, thus the current stack is again $a$. Node $U_5$ receives protocol $a$ and it is only able to decapsulate protocol $b$ from $a$; however, there is no $b$ encapsulated in $a$ at this stage. This path is said to be \textit{not feasible}. The same problem would appear if a node receives a protocol while it is only able to handle another one. In contrast, a path is \textit{feasible} if at each stage, the current node can handle the current protocol stack, and the destination node receives a packet with no encapsulated protocols.
Figure~\ref{fig:with} illustrates a feasible path that involves a loop. This loop is necessary to collect the required encapsulations to cross the subpath $U_4\dots U_6$. This example illustrates some specificities of feasible paths. There is not always a feasible path between two nodes, even if the graph is connected. Moreover, the shortest feasible path (if any) can involve loops. And optimal feasible paths do not exhibit optimal substructure (a subpath of a feasible path is not necessarily feasible, since the protocol stacks at the start and the end of the subpath may be different). In practice, the protocol stack is the sequence of nested headers of a packet that underwent multiple encapsulations.

The problem we address is to compute the shortest feasible path (if any) between each pair of nodes in a distributed way, i.e, the \textit{All Pairs Shortest paths} (APSP) problem in a network involving adaptation functions. Our algorithm builds routing tables allowing a packet with some protocol stack to reach the destination through the shortest feasible path.
The main motivation of this work is to propose a generic algorithm and a basic protocol for routing. Given a network with different protocols and some adaptation functions on some nodes, the goal is to compute the shortest path between any pair of nodes. Tunnels should be automatically established during the routing table construction. This can have important applications on several fields: IPv4/IPv6 interoperability, VPNs, secured tunneling, emulation of lower layer protocols over IP~\cite{RFC3985}, etc.

\begin{figure*}[h!] 
	\centering
	\subfloat[\scriptsize Unfeasible path without loop.]{%
		\includegraphics[width=0.47\textwidth]{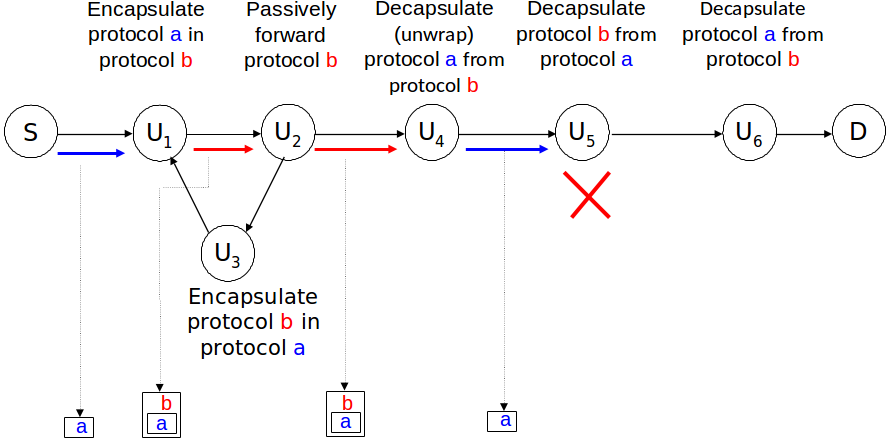}
		\label{fig:without}}
	\subfloat[\scriptsize Shortest feasible path with loop.]{%
		\includegraphics[width=0.47\textwidth]{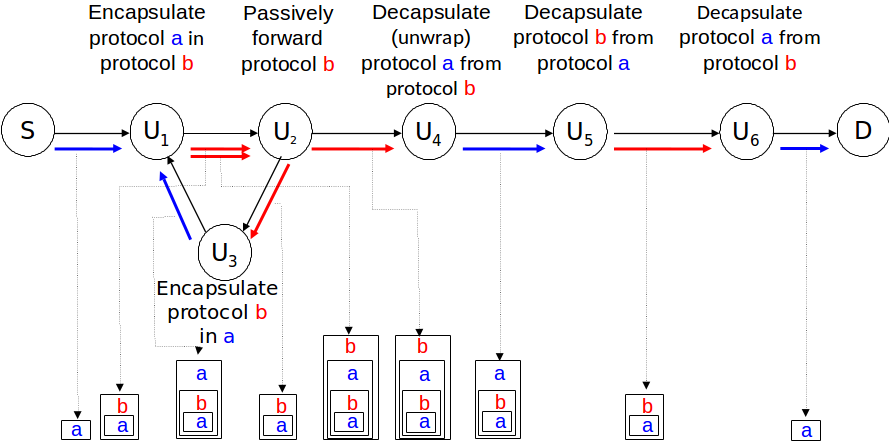}
		\label{fig:with}}
	\caption{A network requiring a shortest path with loop.}
	\label{fig:with_without} 
\end{figure*}

Note that our problem is different from path computation in multi-layer networks where the number of layers is fixed~\cite{jiang2018class}, while our problem is to resolve the APSP in a network where the number of needed layers is unknown, the layers being induced by the node encapsulation and decapsulation capabilities. Moreover, in our setting, even two layers with the same stack height may not be comparable, since two stacks such as $abb$ or $aba$ cannot be routed in the same way, for example.

\subsection{Related Work}

\subsubsection{Centralized approach}
In~\cite{K09}, The authors prove that the path feasibility problem under bandwidth constraint is $\mathsf{NP}$-hard. They propose an exponential algorithm using a  breadth-first search approach to solve the problem. The authors of~\cite{Lamali2013} prove that the problem is polynomial\footnote{More precisely, the decision version, i.e., deciding if a feasible path exists between two nodes, is polynomial. However, as stated in Section~\ref{sec:diff_app}, explicitly computing the path can be superpolynomial.} if the bandwidth constraint is relaxed. They model the network as a Push-Down Automaton (PDA), and use tools from language theory to find the shortest feasible path between two nodes. In~\cite{lamali2018algo}, they  widely generalize the previous algorithms, and prove that the feasible path problem under bandwidth constraint remains $\mathsf{NP}$-hard even with only two protocols and on symmetric directed graphs. The authors of~\cite{Iqbal2015} propose a matrix based model and an algorithm to compute feasible paths. However, the algorithm takes exponential time, and can only compute loopless feasible paths.
Unfortunately, the language theoretic approach seems unlikely to be distributed. Some works have been devoted to distributed PDAs (see~\cite{csuhaj2000parallel} for example). However, in these works, the nodes are themselves PDAs. In contrast, our model is equivalent to a PDA where the states collaborate in a distributed way to find an accepting transition sequence.

\subsubsection{Distributed algorithm related work}
In distributed literature, adaptation functions are generally not considered. The distributed Bellman-Ford algorithm (on which relies RIP) solves the APSP problem in $O(n)$ rounds in a synchronized setting if the message size is unbounded (e.g., the $\mathsf{LOCAL}$ model~\cite{peleg2000distributed}).
However, the APSP problem mostly received attention in the $\mathsf{CONGEST}$ model, where the message size is restricted to be logarithmic~\cite{abboud2016near,Elkin17}.
Unfortunately, these works cannot apply in our context since they do not take into account the adaptation functions and the interoperability issues.
\subsubsection{Networking related work}
Feasible path computation is a major challenge in networking. It underlies several technical problems: establishing tunnels, unifying control planes, etc. IPsec is a classical example. The endpoints of IPsec tunnels are often manually determined. The Tunnel Endpoints Discovery mechanism~\cite{bang2010study} allows to automatically determine IPsec tunnel endpoints. However, this can be only done through a single path. This path corresponds to the shortest path without encryption.
IPv4/IPv6 interoperability is mainly achieved through conversion of protocols or by tunneling. In the first case, an IPv4 (resp. IPv6) header is converted into an IPv6 (resp. IPv4) header. This operation is costly, and it requires to know the nodes able to perform the conversion. NAT-TP~\cite{RFC2766} uses this approach. Some other mechanisms (6over4~\cite{rfc2529}, 6to4~\cite{rfc3056}, etc.) use tunneling. 6to4 is referenced as automatic tunneling. Any 6to4 encapsulation capable router has an anycast address~\cite{rfc3068}. Discovering the tunnel endpoints is done by sending a packet to the anycast address. However, the end-hosts should know that they need a tunnel. Moreover, this mechanism is specific to IPv4/IPv6 interoperability, and cannot handle nested tunnels with arbitrary protocols. While our goal is to design a generic protocol without these limitations. Some other works attempt to automate a part of the process. For example, the TSP protocol~\cite{rfc5572} proposes the (automatic) negotiation of the tunnel parameters (e.g., keep alive duration) between the endpoints. However, no algorithm or process is provided to select the optimal (and may be nested) tunnel endpoints. The ISATAP protocol~\cite{rfc5579} proposes to select the endpoints from a predetermined list via a tunnel Broker. But, again, no algorithm is provided to automatically compute the tunnel endpoints. Moreover, these mechanisms are unable to handle an arbitrary number of nested tunnels. Fully automatic tunneling is still an open problem. The RFC~7059~\cite{rfc7059} proposes a comparison between the different IPv4/IPv6 tunneling mechanisms.

\subsection{Our approach}
\label{sec:our}
We propose to generalize the distributed Bellman-Ford algorithm. Since the path to follow by a packet depends upon the destination, but also upon the packet protocol stack, the distance-vector should be replaced by a stack-vector. Each node advises its neighbors that it can reach some destination, with some protocol stack, and at some cost. The rows of the routing table of each node should contain a next hop per destination and per protocol stack.
The termination of such an algorithm is an important issue. Since feasible paths may involve loops, the same destination can be reached with different protocol stack heights. The question is to know at what stack height to stop. We resolve the problem by showing a bound on the stack height of any shortest feasible path. The convergence speed of the algorithm depends upon the length of the shortest feasible paths. 
Our algorithm is generic: it can handle any number of protocols, and it does not limit the number of nested tunnels to find a feasible path.

\section{Model and definitions}
\label{sec:model}
This section describes the mathematical model used in this paper and formalizes the notion of path feasibility. We use the same notations and definitions as in~\cite{lamali2018algo}.
%
%

\subsection{Network model}
A network is modeled as a $4$-tuple $\N=(\Gg,\alphabet,\mathcal{F},w)$. The network topology is represented by a symmetric directed graph  $\Gg=(\V,\E)$. 
Each node in $ \V$ corresponds to a router, and each pair $(U,V) \in \E$ is a unidirectional link.
The number of nodes is denoted by $|\V|=n$, and the number of links is denoted by $|\E|=m$.
The set of protocols available in the network is denoted by $\alphabet=\{a,b, \dots \}$.  Its cardinality is denoted by $ |\alphabet|=\ca$. 
Each node $U\in\nodes$ has a set $\AD(U)$ of adaptation functions that it can perform. 
These functions are: 
\begin{itemize}
	\item \textit{Conversion}: A protocol $a$ is converted into a protocol $b$ without any change of the possible encapsulated protocols. It is denoted by $(a\rightarrow b)$ (example: IPv4/IPv6 conversion through NAT-PT). Observe that a classical retransmission without protocol change is a special case of conversion. It is denoted by $(a\rightarrow a)$.
	\item \textit{Encapsulation}: A protocol $a$ is encapsulated in  $b$. It is denoted by $(a\rightarrow ab)$ (e.g., IPv4/IPv6 encapsulation).
	\item \textit{Decapsulation}: A protocol $a$ is decapsulated from  $b$. It is denoted by $\overline{(a\rightarrow ab)}$.
\end{itemize}
The set of all adaptation functions available in the network is denoted by $\AD$. We denote the set of protocols that a node can receive by $In(U)$, and the set of protocols that a node can send as $Out(U)$.

Finally,  performing an adaptation function on a node $U$ has a cost defined by the function $w : \V\times\AD\times \V \rightarrow \Re_+$. The value $w(U,f,V)$ (where $U,V\in \V$ and $f\in\AD(U)$) corresponds to the cost of using link $(U,V)$ with adaptation function $f$ on node $U$. Hence, function $w$  represents any additive metric associated with both links and adaptation functions.  So the cost of a path is the sum of the costs of each triple $(U,f,V)$ involved in it in the network. For example, if one wants to minimize the number of encapsulations in the path, the cost function should be $w(U,f,V)=1$ if $f$ is an encapsulation and $0$ otherwise.

\subsection{The protocol stack}

A sequence of adaptation functions induces a protocol stack. For example, the sequence 
$(a\rightarrow a)(a\rightarrow ab)(b\rightarrow b)$
induces the stack $ab$ (from bottom to top). For each position $i$ in a path, $H_i$ denotes the protocol stack at this position, i.e.,  the protocol stack induced by $f_0\dots f_i$, and $h_i$ denotes the protocol stack height. The protocol at the top of a stack $\pile$ is denoted by $\topp{\pile}$ and the protocol just below $\topp{H}$ in the stack is denoted by $\toppp{H}$. The height of a stack $H$ is denoted by $\haut{H}$. The ``forbidden'' stack is denoted by $\emptyset$ (note that it should not be confused with the empty word $\epsilon$).

More formally, let $f$ be an adaptation function, and let $H$ be a stack and $H=H'.\topp{H}$\footnote{The notation ``$.$'' stands for a simple concatenation. For example, if $H=abab$ then $H.b=ababb$. }, where $H'$ is eventually empty.
We will also denote by $f$ the function taking as argument a stack and performing the adaptation function on this stack:
\begin{itemize}
\item if $f=(x\rightarrow y)$ and $\topp{H}=x$, then  $f(H)=H'.y$\\
\item if $f=(x\rightarrow xy)$ and $\topp{H}=x$, then  $f(H)=H.y$\\
\item if $f=\overline{(x\rightarrow xy)}$ and $\topp{H}=y$ and $\toppp{H}=x$, then  $f(H)=H'$\\
\item $f(H)=\emptyset$ otherwise.
\end{itemize}
We also denote by $\bar{f}$ the reverse function of $f$, i.e., if $H'=f(H)$ and $H'\neq \emptyset$ then $\bar{f}(H')=H$. Note that if $f$ is an encapsulation, then $\bar{f}$ is the corresponding decapsulation.
$f(H)=\emptyset$ means that the adaptation function $f$ cannot handle the stack $H$ (e.g., $f$ encapsulates $a$ in $b$ while $\topp{H}\neq a$). Note that $f(\emptyset)=\emptyset$ for any function in our context.

Thus, the protocol stack $H_i$ induced by a sequence of adaptation functions $f_0\dots f_i$, is recursively defined as following:
\begin{itemize}
\item $H_0=x$ if $f_0=(x\rightarrow x)$ and $H_i=f_i(H_{i-1})$
\end{itemize}

\subsection{Path feasibility}
\label{sec:path_feas}
In our context, a path  should take into account the adaptation function capabilities of the network nodes. A path should contain the list of adaptation functions involved in it. Thus, a path from node $S$ to node $D$ in a network $\N$ is a sequence of nodes and adaptation functions $Sf_0U_1f_1U_2f_2\dots U_\ell f_\ell D$ where each $U_i$, $i=1,\ldots,\ell$, is a node, and each $f_i$ is an adaptation function \footnote{The adaptation function $f_0$ is dummy. If node $S$ emits packets of protocol $a$, then $f_0$ is denoted by $(a \to a)$ by convention.}. When needed, a directed path in the graph without taking into account the adaptation functions (i.e., just a sequence of nodes) is referred to as a \textit{classical path}.
\begin{definition}
	A path $\P=Sf_0U_1f_1U_2f_2\dots U_\ell f_\ell D$ is \textit{feasible} if and only if: 
	\begin{enumerate}
		\item The sequence $SU_1U_2\dots U_\ell D$ is a classical path in $\Gg=(\nodes,\E)$ and each $f_i \in \AD(U_i)$;
		\item $H_\ell=x$ and $x\in In(D)$.
        \end{enumerate}
\end{definition}
Actually, the protocol sequences of feasible paths can be characterized as a context-free language~\cite{Lamali2013}. In Figure~\ref{fig:without}, the depicted path is: 
$$\small
S(a\rightarrow a)U_1(a\rightarrow ab)U_2(b\rightarrow b)U_4\overline{(a\rightarrow ab)}U_5\overline{(b\rightarrow ba)}U_6fD.
$$
It cannot be feasible for any adaptation function $f\in\mathcal{F}(U_6)$ since $\overline{(a\rightarrow ab)}$ appears before any encapsulation of protocol $b$. In contrast, Figure~\ref{fig:with} depicts the feasible path:
\begin{equation*}
\begin{split}
&S(a\rightarrow a)U_1(a\rightarrow ab)U_2(b\rightarrow b)U_3(b\rightarrow ba)U_1(a\rightarrow ab)\\
&U_2(b\rightarrow b)U_4\overline{(a\rightarrow ab)}U_5\overline{(b\rightarrow ba)}U_6\overline{(a\rightarrow ab)}D.\\
\end{split}
\end{equation*}
The corresponding protocol stacks are below the links in Figure~\ref{fig:with_without}.
\subsection{Problem formalization}
The cost of a path $\P=Sf_0U_1f_1U_2f_2\dots U_\ell f_\ell D$ from node $S$ to node $D$ is defined as $w(\P) \overset{def}{=}\sum_{i=0}^{\ell}w(U_i,f_i,U_{i+1})$ where $S=U_0$ and $D=U_{\ell+1}$. An optimal (i.e., shortest) feasible path $\P$ between two nodes is a feasible path that minimizes $w(\P)$.
When needed, we refer to the number of hops in a path $\P$ by $|\P|$. Note that $|\P|$ may be different from $w(\P)$.

The problem we aim to solve is the following: compute the routing tables at each node, such that for each pair of nodes $U$ and $V$, a packet from $U$ to $V$ following these routing tables follows an optimal path from $U$ to $V$. This corresponds to an APSP in a network involving adaptation functions.
Note that there is not always a feasible path between two nodes even if the underlying graph $\Gg$ is strongly connected.



\section{Stack-vector routing algorithm}
\label{sec:algorithm}

\subsection{Different approaches to distributed routing}
\label{sec:diff_app}
There are two main approaches in distributed routing: link-state routing and distance/path vector routing. The first one consists in spreading the whole topology then computing locally the path in a centralized way. However, this approach has a drawback: computing the whole shortest path at a node is not enough to route the packets. A node should know the next neighbor at each stage of the route. This is not possible unless the header of the packet contains the whole path to follow. 
The second approach is to use a Bellman-Ford algorithm (as RIP) where each node shares its current routing table with its neighbors. We cannot use this approach since the path leading to the destination depends upon the protocol stack of the packet to route. Storing the cost and the destination in the routing table is not enough.

Path vector routing consists in sharing not only the destination and the cost, but also the whole path to reach the destination. The main goal to do so is to avoid loops (if a node receives a path where it already appears, then it discards it). BGP uses this approach. This would be a possible solution for our problem, since a complete feasible path (the node sequence and the adaptation function sequence) characterizes the route to take for a packet. However, the shortest feasible path can be superpolynomial~\cite{lamali2018algo}. A lower bound of $\Omega((n/\ca )^\ca)$ is known for an arbitrary number of protocols $\ca$. Even for two protocols, there are shortest paths of length at least $\Omega(\sqrt{n}2^{\Omega(\sqrt{n})})$. The best available upper bound is $2^{O(\ca^2n^2)}-1$ (see \cite{lamali2018algo} for more details). Using a path vector protocol would lead to a superpolynomial message length.
We opted for a \textit{stack vector algorithm} (by analogy with distance vector protocols), where the destination, the cost and the protocol stack are stored in the routing table.

\subsection{The algorithm}
The initialization algorithm (Algorithm~\ref{algo:init}) allows each node $U$ to share with all its neighbors $\voisin{U}$ the set of protocols $In(U)$ that it can receive. Sending a message $(U,\pile,c)$ means that the sender can reach destination $U$ at cost $c$ if it receives a packet with protocol stack $\pile$. Thus, the initialization phase consists of informing its neighbors that it can receive any packet with protocol $x\in In(U)$ (without nested encapsulations) at cost $0$.

\begin{algorithm}
	\begin{algorithmic}[1]
		\ForAll{$x\in In(U)$}
			\ForAll{$V\in \voisin{U}$}
				\State $\pile \gets x$
				\State $\send{U}{\pile}{0}$ to $V$
			\EndFor
		\EndFor
	\end{algorithmic}
	\caption{Initialization algorithm of node $U$}
	\label{algo:init}
\end{algorithm}

Each row of the routing table $\tr$ of each node is a $5$-tuple $(D,\pile,c,V,f)$ where $D$ is the destination to reach, $\pile$ is the protocol stack needed to reach $D$, $c$ is the cost of the remaining path to reach the destination, node $V$ is the next neighbor (next hop) to reach the destination, and $f$ is the adaptation function to perform at the current node. 
The row is indexed by the pair $(D,\pile)$. Thus, $\tr(D,\pile)$ returns the row corresponding to destination $D$, and stack $\pile$. When a node tries to add a new row $(D,\pile,c,V,f)$ to its routing table $\tr$, it checks if the tuple $(D,\pile)$ is already in its routing table. If it is not, then it inserts the new row to the table. Otherwise, it compares the cost of the new route with the old one. The table is updated by replacing the old route by the new one if the new cost is lower. This step is done according to Algorithm~\ref{algo:add_row}.

\begin{algorithm}
	\begin{algorithmic}[1]
		\Require A row $(D,\pile,c,V,f)$
		
		\If{$(D,\pile)\notin\tr$} 
		\State Add $(D,\pile,V,c,f)$ to $\tr$
		\ElsIf{$\tr(D,\pile).cost>c$}
		\State $\tr(D,\pile).cost\gets c$
        \State $\tr(D,\pile).next\_hop\gets V$
        \State $\tr(D,\pile).function\gets f$
		\EndIf
	\end{algorithmic}
	\caption{Add a row to a routing table $\tr$ of node $U$}
	\label{algo:add_row}
\end{algorithm}

Algorithm~\ref{algo:construct_table} is the main algorithm. When receiving a new message $(D,H,c)$ from a neighbor $V$, node $U$ determines which of its adaptation functions it can apply to stack $H$ (lines 3-5). Then it computes the new cost of the remaining path by adding the cost of $(U,f,V)$ to the old cost (line~6). It tries to add the new row corresponding to the message and the chosen adaptation function to its routing table $\tr$, according to Algorithm~\ref{algo:add_row} (line~7). If the row corresponds to a new route in the table, then it sends the new message $(D,H,c)$ to its neighbors, indicating that it can reach the destination with the new stack $H$.
Observe that if $f$ is a decapsulation $\overline{(a\rightarrow ab)}$, and $U$ can reach $D$ with some protocol stack $H$ such that $\topp{H}=a$, then $U$ can handle a packet with stack $H.b$ to reach $D$. Similarly, If $f$ is a conversion $(a\rightarrow b$), and if $U$ can reach $D$ with some protocol stack $H$ such that $\topp{H}=b$ (i.e., $H=H'.b$ for some stack $H'$), then $U$ can handle a packet with stack $H'a$ to reach $D$. Thus, if node $U$ can reach a destination $D$ with some stack $H$, for each $f\in\AD(U)$, it should apply the reverse adaptation function $\bar{f}$ (line~4) to $H$ before sending a message to the other nodes.
Note that only stacks with heights less than or equal to $\ca n^2$ are kept and shared (line~5). In the next section, we prove that this height is sufficient to reach any destination with a shortest feasible path if there exists one.

\begin{algorithm}
	\begin{algorithmic}[1]
		\Loop
					\State Receive $({D,\pile,c})$ from $V$
					
		\ForAll{$f \in \AD(U)$}
        	\State $\pile\gets\bar{f}(\pile)$
			 \If{$\pile\neq\emptyset$ and $\haut{\pile}\leq \ca n^2$}
				 \State $c\gets c+w(U,f,V)$
                 \State Add row $(D,\pile,c,V,f)$ to $\tr$
                 \State \verb+#+ \textit{according to Algorithm~\ref{algo:add_row}}
				 \If{$\tr$ has been modified}
                 	\ForAll{$W\in \voisin{U}$}
                    	\State Send $(D,H,c)$ to W
                    \EndFor
                 \EndIf
			\EndIf
		\EndFor
		\EndLoop
	\end{algorithmic}
	\caption{Routing table construction algorithm of node $U$}
	\label{algo:construct_table}
\end{algorithm}

We assume that the packets to route are of the form (destination, protocol stack, payload); for example, $(D,H,data)$. Once all routing tables are computed, if node $U$ receives a packet with protocol stack $H$ and destination $D$, it first searches for tuple $(D,H)$ in its routing table. If there is no corresponding row, then node $U$ has no route for the destination with the received protocol stack. Otherwise, if the corresponding row is $(D,H,c,V,f)$, then it sends the packet $(D,f(H),data)$ to $V$. The routing procedure is illustrated by Algorithm~\ref{algo:route_packet}.

\begin{algorithm}
	\begin{algorithmic}[1]
		\State Receive a packet $(D,\pile,data)$
		\If{$(D,\pile)\notin\tr$} 
        	\State No route, discard the packet
		\Else
        	\State Let $(D,\pile,c,V,f)\in\tr$
            \State Send the packet $(D,f(\pile),data)$ to $V$.
		\EndIf
	\end{algorithmic}
	\caption{Routing a packet.}
	\label{algo:route_packet}
\end{algorithm}

\subsection{Correctness and complexity}

\subsubsection{Stack height upper bound and message size}
First, we prove that for any shortest feasible path, the maximum stack height reached along the path is polynomially bounded.
\begin{lemma}
	\label{cor:max}
	The shortest feasible path (if any) between two nodes reaches a maximum stack height of at most $\ca n^2$ protocols.
\end{lemma}
\begin{proof}
	See Appendix~A.
\end{proof}

\begin{corollary}
	The maximum message size is in $O(\ca \log \ca\  n^2)$.
\end{corollary}
\begin{proof}
A message contains node identifiers, that are in $O(\log n)$ if the number of nodes is $n$. A protocol stack, its height is at most $\ca n^2$ according to Lemma~\ref{cor:max}. The identifier of a protocol is in $O(\log \ca)$ for a number $\ca$ of protocols. We assume that the cost is bounded.
\end{proof}

\subsubsection{Convergence and correctness} We will prove that in a synchronized setting (for example the $\mathsf{LOCAL}$ model in distributed computing~\cite{peleg2000distributed}) where all nodes receive messages at the same time $t$, then process them and send them at the same time $t+1$, Algorithm~\ref{algo:construct_table} correctly converges in polynomial time\footnote{Time is measured in number of synchronized rounds.} according to the network size and its diameter.

We define the diameter of a network $\N$, denoted by $\diam \N$, as the length (in number of hops) of the shortest feasible path that maximizes the number of hops:
$$\diam \N \overset{def}{=}\max_{\P \text{ shortest feasible}} |\P|\ .$$

\begin{proposition}
Algorithm~\ref{algo:construct_table} computes the correct routing tables after $O(\ca n^2\ \diam \N)$ rounds.
\label{prop:conv}
\end{proposition}
\begin{proof}
See Appendix~B.
\end{proof}

\section{A proto-protocol}
\label{sec:proto}
\subsection{Practical limitations and algorithm adaptation}
\label{sec:limit}
The main issue for the implementation of the proposed algorithm is that the stack height and the message length bounds are too large, even if they are polynomial. For example, in a network of $100$ nodes with only $2$ protocols, the stack height may reach $2\times10^4$ protocols. This is not sustainable for real applications. Moreover, for a stack height $h$, there are $\ca^h$ possible stacks, i.e., in the worst case, there may be $n\left(\frac{1-\ca^{\ca n^2 +1}}{1-\ca}-1\right)$ rows in a routing table. This may induce an exponential number of exchanged messages. Finally, as demonstrated by the authors of~\cite{lamali2018algo}, $\diam \N$ can be superpolynomial in the size of the network $\N$. These bounds are tight, i.e., it is possible to exhibit a network where the shortest feasible path reaches a stack height of $\ca n^2$. Thus, these limitations are not due to our algorithm but are inherent to the problem.

However, the simulations performed in \cite{lamali2018algo} show that such instances are extremely unlikely to appear. Thus, we propose to set the maximum stack height as a parameter of the protocol. A small value of this parameter is enough to compute the shortest path in most cases. We should bound the maximum stack height by a constant $h_{\max}$ for any real implementation. Thus, the condition $h(H)\leq \ca n^2$ in Algorithm~\ref{algo:construct_table} (line~5) should be replaced by $h(H)\leq h_{\max}$. In such a case, the maximum number of different stacks of height $h_{\max}$ would be $\ca^{h_{\max}}$, and the maximum number of rows in a routing table would be $n\left(\frac{1-\ca^{h_{\max}+1}}{1-\ca}-1\right)$, which is linear in $n$ and polynomial in $\ca$.

\subsection{Routing \textit{proto}-protocol specification}
The main requirements are that the nodes and the protocols must have a unique identifier. For example, 1 byte for the protocol identifiers and 16 bytes for the node identifiers. The routers involved in the routing protocol should have a specific multicast address.
The routing protocol messages that advise a route must contain: the destination of the route, the protocol stack needed for a packet in order to reach the destination, and the cost of the route. In addition, some classical parameters can be exchanged: keep alive duration, emitted message timestamps, etc. The maximum stack height should be set at the start of the process.
The routing table must contain entries indexed by the destination and the protocol stack of a received packet. Each entry must contain the cost, the next hop, and the adaptation function to perform before sending the packet to the next hop.
\begin{figure}
\centering
\begin{tikzpicture}[scale=0.6, every node/.style={scale=0.6},
node distance = 0pt,
    BC/.style = {
        decorate,
        decoration={calligraphic brace, amplitude=1.2mm,
        pre =moveto, pre  length=0.75pt,
        post=moveto, post length=0.75pt,
        raise=1mm,
        #1},
        very thick,
        pen colour={red}
                  },
  BC/.default = mirror,
    LN/.style = {inner xsep=4pt, outer sep=0pt},
                        ]
\matrix (m) [matrix of nodes, inner sep=0pt,
             nodes={text depth=0.8ex, text height=1em, minimum width=25ex,
                    inner ysep=1pt, inner xsep=4pt, outer sep=0pt, anchor=west},
             nodes in empty cells,
             column sep=-\pgflinewidth,
             row sep= -\pgflinewidth,
             ]
{
    Destination             \\
    Destination (cont.)         \\
    Source          \\
    Source (cont.)     \\
    Prot. 1          \\
    Prot. 2   \\
    \dots      \\
     Prot. $i$ \\
     Header length of prot. 1 \\
     Prot. 1 header \\
     \dots \\
     Header length of prot. 2 \\
     Prot. 2 header \\
     \dots \\
};

\draw           (m-1-1.north west) -- (m-1-1.north east);
\draw           (m-1-1.south west) -- (m-1-1.south east);
\draw           (m-2-1.south west) -- (m-2-1.south east);
\draw           (m-3-1.south west) -- (m-3-1.south east);
\draw           (m-4-1.south west) -- (m-4-1.south east);
\draw           (m-5-1.south west) -- (m-5-1.south east);
\draw           (m-6-1.south west) -- (m-6-1.south east);
\draw           (m-7-1.south west) -- (m-7-1.south east);
\draw           (m-8-1.south west) -- (m-8-1.south east);
\draw           (m-9-1.south west) -- (m-9-1.south east);
\draw           (m-10-1.south west) -- (m-10-1.south east);
\draw           (m-11-1.south west) -- (m-11-1.south east);
\draw           (m-12-1.south west) -- (m-12-1.south east);
\draw           (m-13-1.south west) -- (m-13-1.south east);

\draw           (m-1-1.north west) -- (m-6-1.south west);
\draw [dotted] (m-7-1.north west) -- (m-7-1.south west);
\draw           (m-8-1.north west) -- (m-10-1.south west);
\draw [dotted] (m-11-1.north west) -- (m-11-1.south west);
\draw           (m-12-1.north west) -- (m-13-1.south west);
\draw [dotted] (m-14-1.north west) -- (m-14-1.south west);

\draw           (m-1-1.north east) -- (m-6-1.south east);
\draw [dotted] (m-7-1.north east) -- (m-7-1.south east);
\draw           (m-8-1.north east) -- (m-10-1.south east);
\draw [dotted] (m-11-1.north east) -- (m-11-1.south east);
\draw           (m-12-1.north east) -- (m-13-1.south east);
\draw [dotted] (m-14-1.north east) -- (m-14-1.south east);

\draw[BC={}]    (m-5-1.north -| m.east) --
                    node[right=3mm] {Protocol identifier stack}
                (m-8-1.south -| m.east);
\draw[BC={}]    (m-9-1.north -| m.east) --
                    node[right=3mm] {Protocol header stack}
                (m-14-1.south -| m.east);
    \end{tikzpicture}
    \caption{A header of a packet to route.}
   \label{fig:header}
\end{figure}
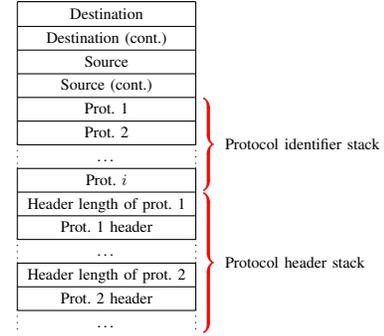

\subsection{Routing a packet}
The packet to route contains a stack of headers of different protocols. It can be seen as a meta-header. It must contain:
\begin{itemize}
	\item The unique identifier of the destination node;
    \item The source (even if it does not impact the routing process, since it is a per destination/stack routing);
	\item The stack height: the current stack height of the packet 
	\item The protocol stack: the stack of identifiers of the protocols corresponding to the nested headers. This will speed up the routing process, since it avoids to access to the whole header stack;
	\item The header stack: each encapsulated header preceded by its length;
	\item The payload of the inner packet.
\end{itemize}
Figure~\ref{fig:header} illustrates such a meta-header.

\section{Simulations}
\label{sec:simuls}
In order to evaluate the efficiency of our algorithm, we performed simulations with different parameters.

\subsection{Simulation methodology}
\label{sec:methodology}
All the networks used in the simulations are generated according to the following steps:
\begin{enumerate}
\item We generate a random undirected graph of a given size according to a preferential attachment mechanism (the Barab\`asi\--Albert model~\cite{Barabasi1999}), where each added node is attached to $5$ existing nodes;
\item The graph is then converted into a symmetric directed graph. Each undirected link is converted into two directed links;
\item For a given number $\ca$ of protocols, there are $3\ca^2$ possible adaptation functions. Each adaptation function is available on a node with a given probability $p$.
\end{enumerate}

The algorithm is implemented in Python 3.4.5, using the NetworkX package\footnote{\url{https://networkx.github.io/}}. The implementation is done in a distributed fashion: each node is simulated by a thread, and a directed link $(U,V)$ is implemented as a queue where $U$ can only write, and $V$ can only read. The simulations were performed on a multi-core server with CPUs 1.59GHz.

The input parameters of the algorithm are: the number of nodes, the probability $p$ of availability of an adaptation function, the number of protocols, and the maximum stack height $h_{\max}$. 
The main output results are the convergence time, and the percentage of times where the algorithm finds the shortest path between the network extremities. Note that, if the algorithm does not find the shortest feasible path, it may be because there is no feasible path in the given network. The probability of existence of a feasible path according to different parameters can be found in~\cite{lamali2018algo}.
All the results are averaged over 1000 runs.

\subsection{Convergence time}
\begin{figure*} 
	\centering
	\subfloat[\scriptsize Convergence time according to network size.]{%
		\includegraphics[width=0.33\linewidth]{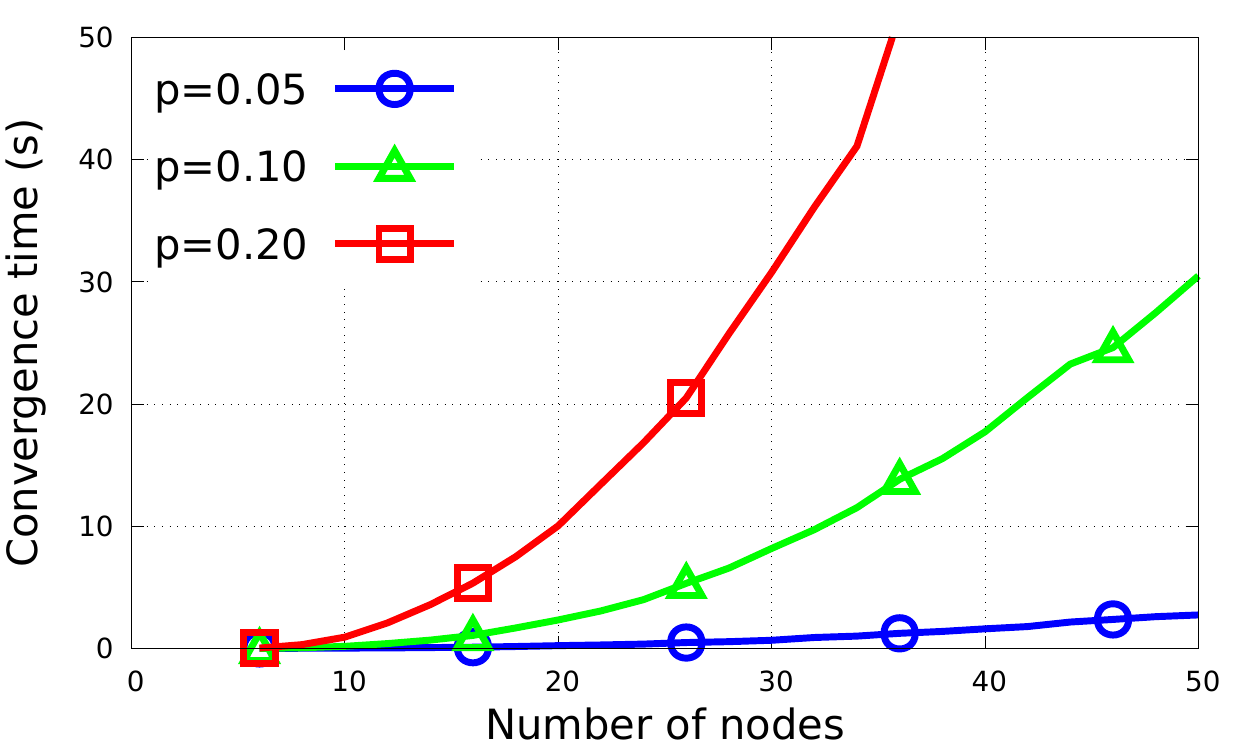}
		\label{fig:taille_vs_temps}}
	\subfloat[\scriptsize Convergence time according to probability $p$.]{%
		\includegraphics[width=0.33\linewidth]{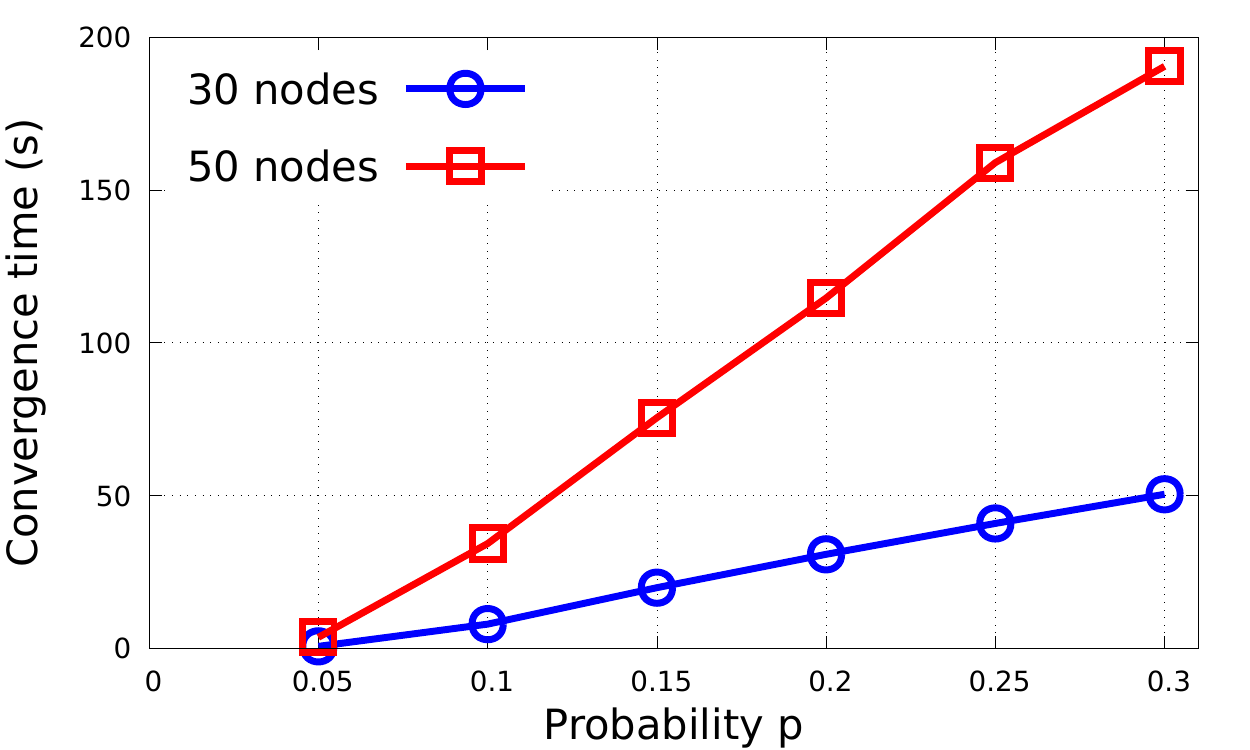}
		\label{fig:p_vs_temps}}
	\subfloat[\scriptsize Convergence time according to $h_{\max}$. Note the logarithmic scale on the $y$-axis.]{%
		\includegraphics[width=0.33\linewidth]{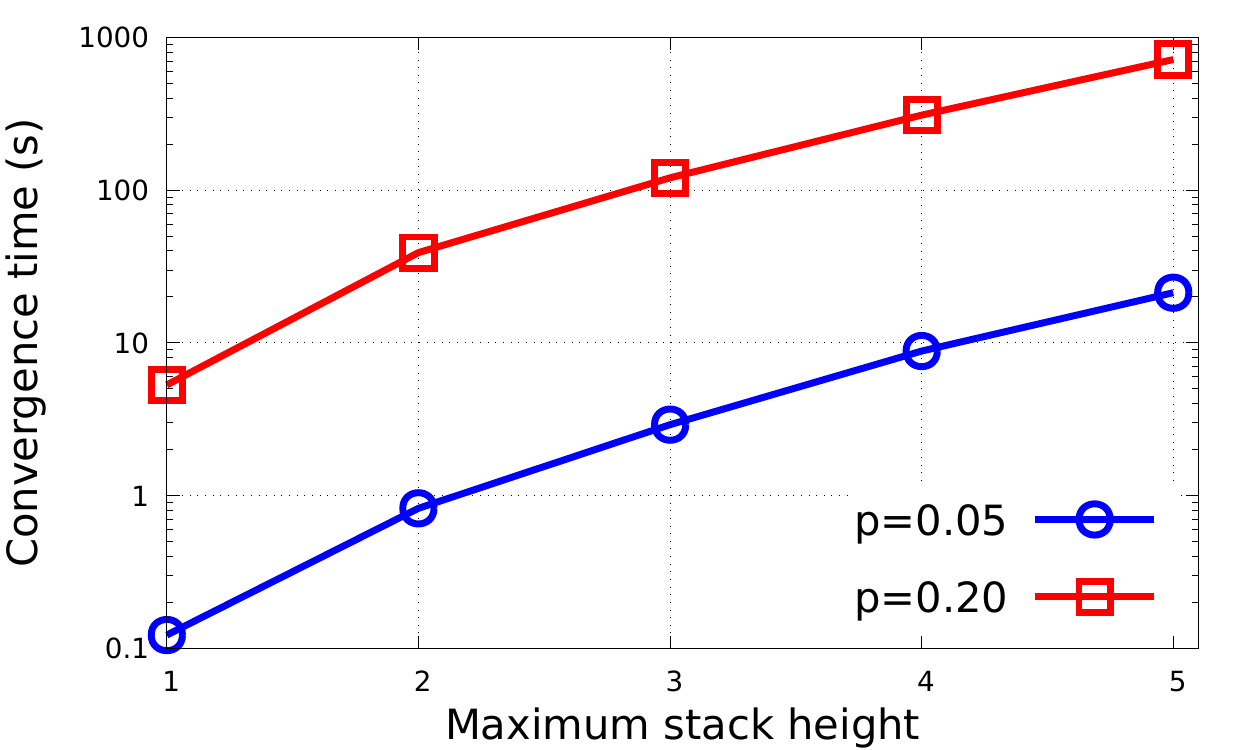}
		\label{fig:hauteur_vs_temps}}
	\caption{The convergence time of Algorithm~\ref{algo:construct_table} according to different parameters.}
	\label{fig:convergence_time} 
\end{figure*}

Figure~\ref{fig:convergence_time} shows the convergence time of the algorithm according to different parameters.  Figure~\ref{fig:taille_vs_temps} shows the convergence time according to the network size (number of nodes) in three cases: $p=0.05$, $p=0.10$, and $p=0.20$.  The maximum stack height $h_{\max}$ is set to $3$, and the number of protocols is set to $2$. It appears that the convergence time hugely depends on the parameter $p$. For $p=0.05$, the convergence time is $2.74$s in a network of $50$ nodes; while for $p=0.20$, the convergence time is $113$s. However, note that if $p=0.20$, the average number of adaptation functions per node is $2.4$, which is unrealistic, since only a few number of nodes should be able to perform conversions/encapsulations. Actually, Figure~\ref{fig:p_vs_temps} shows the impact of the probability $p$ (and the average number of adaptation functions per node) on the convergence time (with $h_{\max}=3$ and $2$ protocols). For $p=0.05$, the convergence time is around $0.7$s (resp. around $3$s) in a network of $30$ (resp. $50$) nodes; while for $p=0.30$, the convergence time is around $50$s (resp. around $3$min) in a network of $30$ (resp. $50$) nodes. We can see that the processing time hugely depends upon the number of adaptation functions per node.
Figure~\ref{fig:hauteur_vs_temps} shows the impact of the parameter $h_{\max}$ on the convergence time. The number of protocols is set to $2$ and the network has $50$ nodes. Note the logarithmic scale on the $y$-axis. The maximum stack height hugely impacts the convergence time. With $p=0.05$ (resp. $p=0.20$), the convergence time is $0.7$s (resp. around $30$s) if the maximum stack height is set to $2$. However, the convergence time is around $21$s (resp. $11$min) if the maximum stack height is set to $5$.

\subsection{Algorithm efficiency}

\begin{figure} 
	\centering
	\subfloat[\scriptsize $\%$ of found feasible paths according to probability $p$.]{%
		\includegraphics[width=0.48\linewidth]{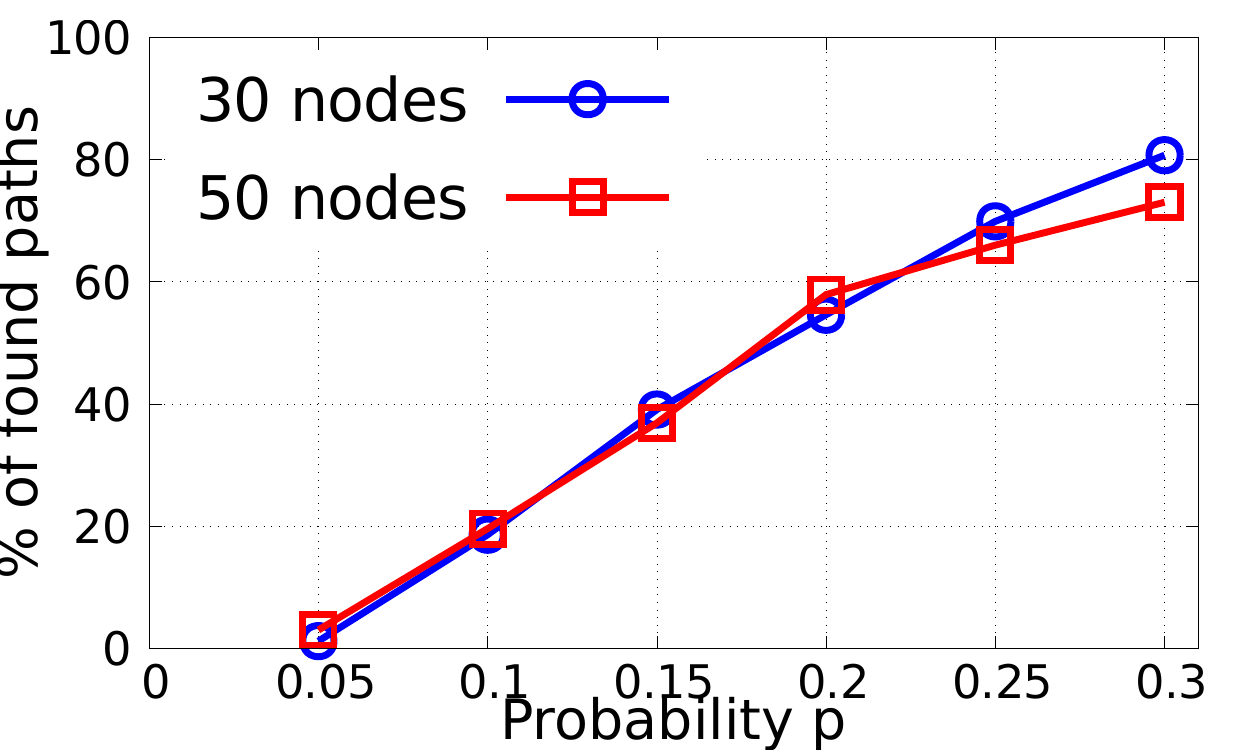}
		\label{fig:p_vs_chemin}}\hfil \hfil
	\subfloat[\scriptsize $\%$ of found feasible paths according to $h_{\max}$.]{%
		\includegraphics[width=0.48\linewidth]{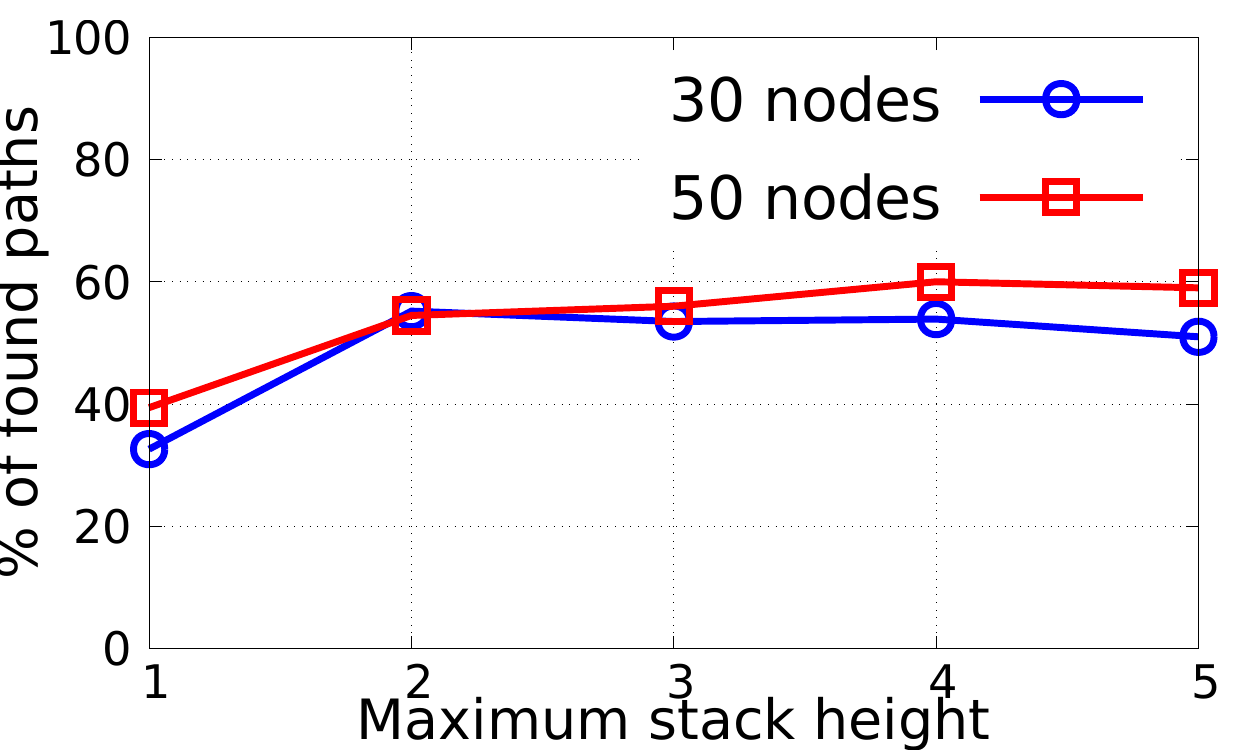}
		\label{fig:hauteur_vs_chemin}}
	\caption{The $\%$ of found feasible paths by Algorithm~\ref{algo:construct_table} according to different parameters.}
	\label{fig:path_found}  
\end{figure}

Figure~\ref{fig:path_found} shows the average number of times where Algorithm~\ref{algo:construct_table} found a shortest feasible path between the network extremities, according to different parameters. For each run, we check if a shortest feasible path is computed between the network extremities. Note that, if there is no path, it can be due to two different reasons: i) there is no feasible path between these nodes; ii) the shortest feasible path reaches a stack height larger than $h_{\max}$.

Figure~\ref{fig:p_vs_chemin} shows the impact of probability $p$ on the number of found paths. For $30$ nodes (resp. $50$ nodes) and $h_{\max}$ set to $3$, the number of runs where a shortest feasible path was found is $1\%$ (resp. $3\%$) if $p=0.05$. While it is $80\%$ (resp. $73\%$) if $p=0.30$. Note that the probability of existence of a feasible path is much smaller if $p=0.05$ than if $p=0.30$.
Figure~\ref{fig:hauteur_vs_chemin} shows the impact of the maximum stack height $h_{\max}$ on the number of found paths. The probability $p$ is set to $0.20$ and there are $2$ protocols. Obviously, the algorithm finds more feasible paths if it is allowed to explore paths with higher stacks. However, the difference between stack heights $3$ and $5$, for example, is negligible. It seems that almost all the shortest feasible paths are of maximum stack height at most $3$ (i.e., two nested tunnels).

These simulations shows that the problem is complex, and the convergence time can be prohibitive, especially if the number of adaptation functions per node is large. However, a convergence time of few minutes is sustainable in a network with infrequent topology changes.

\section{conclusion}
\label{sec:conclusion}
Nowadays, the Internet encompasses several protocols. The interoperability between theses protocols is an important issue, and is ensured thanks to adaptation functions. A path between two nodes may involve different protocols at different stages. The current routing protocols are not able to automatically compute such paths due to the lack of distributed algorithms taking into account the adaptation functions. In this paper, we design the first fully distributed algorithm taking into account these functions. Our algorithm builds at each node a local routing table that allows to route a packet following the optimal path. Moreover, we prove that our algorithm converges polynomially in the size of the network and its diameter, and it uses messages of polynomial size, despite the fact that the shortest feasible paths can be of superpolynomial length. We propose a basic implementation of our algorithm as a stack-vector routing protocol, and we evaluate its efficiency through simulations. We believe that this work can have an important impact on protocols such as TSP, on IPv4/IPv6 interoperability, and on automatic tunneling more generally. 
As a future work, we plan to study the application of our algorithm to secured communication, more precisely to nested encrypted tunnels. In this paper, we assumed that each node had access to the full protocol stack of a packet that it receives. This implies that the data can be encrypted several times but that the headers must not be encrypted. This is not suitable for security reasons. Thus, the main possible improvement to our algorithm would be to adapt it to the case in which all the encapsulated headers are also encrypted. This means that any node has only access to the top (outer) protocol. This would have applications to distributed \textit{onion routing}. 

{\bf Acknowledgement.} The authors would like to thank G\'eraud S\'enizergues for his invaluable help and explanations about the proof of Lemma~\ref{cor:max}. The first author was partially supported by the  H\'ERA project, funded by The French National Research Agency. Grant no.: ANR-18-CE25-0002.

\bibliographystyle{IEEEtran}
\bibliography{automata_infocom}

\section*{Appendix~A}
\label{appA}
\subsection*{Proof of Lemma~\ref{cor:max}}
\begin{proof}
	We will prove that if there is a feasible path from $S$ to $D$ that reaches a maximum stack height $h_{\max}>\ca n^2$, then there is a shorter feasible path from $S$ to $D$ that reaches a maximum stack height $h_{\max}'<h_{\max}$.
	The main ideas of the proof are results from language theory, and are related to the pumping lemma. They can be found in~\cite{ges90} and \cite{amarilli2012proof} for example. 
	
	Let $\P$ be a feasible path of length $\ell$ between two nodes. Suppose that $\P$ reaches a maximum stack height $h_{\max}>\ca n^2$ at some position $j$ (i.e., $h_j=h_{\max}$).
	For each stack height $h$ such that $h_{\max}-\ca n^2\leq h \leq h_{\max}$, let $i_h$ (resp. $k_h$) be the last position before (resp. the first position after) $j$ reaching the stack height $h$. More formally:
	
	\begin{itemize}
		\item $i_h=\max\{i\leq j \mid h_i=h\}$
		\item $k_h=\min\{k\geq j \mid h_k=h\}$
	\end{itemize}
	And let $\langle h \rangle$ be the $3$-tuple $(U_i,U_k,a_i)$ where:
	\begin{itemize}
		\item $U_i$ (resp. $U_k$) is the current node at position $i_h$ (resp. $k_h$),
		\item $a_i$ is the current protocol at position $i_h$ (i.e., $a_i=\topp{H_i}$). Note that $\topp{H_i}=\topp{H_k}$ since $f_i$ is followed by an encapsulation, and $f_k$ follows a decapsulation. 
	\end{itemize}
	For any $h$, there are $\ca n^2$ possible values of $\langle h\rangle$. However, since $h_{max}-\ca n^2\leq h \leq h_{max}$, $h$ can take $\ca n^2+1$ different values. Thus, there are two stack heights $h\neq h'$ such that $\langle h \rangle=\langle h' \rangle$. Let $h'<h$ for convenience. By definition of $i_h$ and $k_h$, the subpath from position $i_h$ to $k_h$ does not involve a protocol stack smaller than $h$. Thus, it depends only on the part of the stack above $h$, on the top of the stack, and on the current node at position $i_h$ (recall that $i_h$ is the \textit{last} position before $j$ reaching stack height $h$, the following positions until $k_h$ involve stacks higher than $h$). Likewise, the subpath between positions $i_{h'}$ and $k_{h'}$ does not involve a stack smaller than $h'$. The part of the stack below height $h$ (resp. $h'$) is transparent to the nodes between $i_h$ and $k_h$ (resp. $i_{h'}$ and $k_{h'}$). The same sequence of adaptation functions and nodes leading from $i_h$ to $k_h$ can also lead  from $i_{h'}$ to $k_{h'}$. It is then possible to shortcut the part of the path between $i_{h'}$ and $i_h$ together with the part from $k_{h}$ to $k_{h'}$.
Thus, the path:
\begin{equation*}
\begin{split}
&\P'= Sf_0U_1f_1\dots U_{i_{h'}-1}f_{i_{h'}-1}U_{i_{h}}f_{i_{h}}\dots U_jf_j\dots U_{k_{h}}f_{k_{h}}\\
&U_{k_{h'}+1}f_{k_{h'}+1}\dots D
\end{split}
\end{equation*}
is feasible, shorter than $\P$, and reaches a maximum stack height $h_{max}-(h-h')$.
\end{proof}

\section*{Appendix~B}
\label{appB}
\subsection*{Proof of Proposition~\ref{prop:conv}}

First we prove the following lemma:
\begin{lemma}
Let $\N$ be a network. If after a number of rounds $t$, all the shortest feasible paths reaching a maximum stack height at most $h_{\max}$ are already computed, then all the shortest feasible paths reaching a maximum stack height at most $h_{\max}+1$ will be computed after $t+O(\diam \N)$ rounds.

\label{lemme:induction}
\end{lemma}

\begin{proof}
Computing a shortest feasible path of maximum stack height $h_{\max}+1$ from shortest feasible paths of maximum stack height $h_{\max}$ can be done in two steps:
\begin{enumerate}
\item The new path is a path of maximum stack height $h_{\max}$ following an encapsulation and followed by a decapsulation;
\item Concatenation of several new paths of maximum stack height at most $h_{\max}+1$.
\end{enumerate}
Suppose that after some rounds, all the shortest feasible paths that reach a maximum stack height $h\leq h_{\max}$ are computed. Let $\P=Sf_0U_1f_1U_2f_2\dots U_\ell f_\ell D$ be a shortest feasible path from node $S$ to $D$ that reaches a maximum stack height $h_j$ at some position $j$, such that $h_j>h_{\max}$. Let $i_h$ (resp. $k_h$) be the last position before (resp. the first position after) $j$ that reaches stack height $h=h_j-h_{\max}$, and follows an encapsulation (resp. is followed by a decapsulation). More formally:
\begin{itemize}
	\item $i_h=\max\{i<j\mid h_i=h_j-h_{\max} \text{ and }h_{i-1}=h_i-1\}$
	\item $i_k=\min\{k>j\mid h_k=h_j-h_{\max} \text{ and }h_{k+1}=h_k-1\}$
\end{itemize}
Then the path $U_{i_h}f_{i_h}\dots U_jf_j\dots f_{k_h-1}U_{k_h}$ is feasible and is already computed, i.e., node $U_{i_h}$ knows that it can reach $U_{k_h}$ by emitting a packet of some protocol $x$, and $(U_{k_h},x,c,U_{i_h+1},f)$ is in its routing table for some cost $c$ and some protocol $x$. 

On the other hand, node $U_{k_h}$ already knows that it can reach node $U_{k_h+1}$ with some protocol $x$, since it received the message $(U_{k_h+1},x,0)$ from $U_{k_h+1}$ at the first round, thanks to Algorithm~\ref{algo:init}. Thus, at round $t$, node $U_{i_h+1}$ knows that it can reach node $U_{k_h+1}$ with some protocol stack. After two rounds, through node $U_{i_h}$, node $U_{i_h-1}$ will know that it can reach node $U_{k_h+1}$ with some protocol $x$.

Now it remains to prove that if $\P_1,\dots,\P_d$ are shortest feasible paths already computed, and the last node of $\P_i$ (call it $D_i$) is the first node of $\P_{i+1}$, then the shortest feasible path $\P'=\P_1\dots\P_d$ (if any) is computed after $O(\diam \N)$ rounds. It means that the first node of $\P_1$ knows that it can reach the last node of $\P_d$ by emitting a suitable protocol. By the same argument, the first node of $\P_1$ will know that it can reach the last node of $\P_d$ after at most $|\P_1,\dots,\P_d|$ rounds. Since, by definition, $|\P_1,\dots,\P_d|\leq \diam \N$, this concludes the proof.
\end{proof}

Recall now Proposition~\ref{prop:conv}:
\setcounter{proposition}{0}
\begin{proposition}
Algorithm~\ref{algo:construct_table} computes the correct routing tables after $O(\ca n^2\ \diam \N)$ rounds.
\label{prop:conv}
\end{proposition}

\begin{proof}
We prove by induction that after $O(\ca n^2\ \diam \N)$ rounds, all the shortest feasible paths that reach a maximum stack height $h\leq h_{\max}$ are computed, i.e., the routing tables are able to route packets following these paths.

{\noindent \bf Basis:} Suppose that there is a feasible path $\P=Sf_0U_1f_1U_2f_2\dots U_\ell f_\ell D$ from $S=U_0$ to $D=U_{\ell+1}$ that keeps stack height of $1$ (i.e., there is neither encapsulations nor decapsulations along the path). 
Since, by definition, $|\P|\leq \diam \N$, after at most $\diam \N$ rounds, the routing table of $S$ contains the row $(D,x,c, U_1,w(\P))$ for $0\leq i\leq \ell$ and some protocol $x\in In(S)$, since $\P$ is feasible.
	
{\noindent \bf Induction}: Suppose that after $t$ rounds, all the shortest feasible paths that reach a maximum stack height at most $h_{\max}$ are computed. By Lemma~\ref{lemme:induction}, all the shortest feasible paths that reach a maximum stack height of $h_{\max}+1$ are computed after $t+O(\diam \N)$ rounds. By Lemma~\ref{cor:max}, the maximum stack height of a feasible shortest path (if any) is $\ca n^2$. This concludes the proof.
\end{proof}

\end{document}